\documentclass[conference,letterpaper]{IEEEtran}

\usepackage{cite}
\usepackage{booktabs}
\usepackage{amsthm}
\newtheorem{lemma}{Lemma}
\newtheorem{theorem}{Theorem}

\usepackage{multirow}
\usepackage{braket}
\usepackage{lineno,xcolor}
\usepackage{amsmath}
\usepackage{amsfonts}%
\usepackage{amssymb}
\usepackage{wasysym}
\usepackage{stmaryrd}
\usepackage{bbm} 
\usepackage{tabularx}
\usepackage{multirow}
\usepackage{mathtools}
\usepackage{enumerate}
\usepackage{subfigure}
\usepackage{pgf}
\usepackage{bm}
\usepackage{multicol}
\usepackage{color}
\usepackage{url}
\usepackage{bibunits}
\usepackage{wrapfig}
\usepackage{sidecap}
\usepackage{soul,xcolor}
\usepackage{eurosym}
\usepackage{mathrsfs}
\usepackage[utf8]{inputenc}
\usepackage{algorithmicx}
\usepackage{algpseudocode}
\usepackage{algorithm}
\usepackage{stmaryrd}
\usepackage{upgreek}
\usepackage[hidelinks]{hyperref}
\usepackage{paralist}
\usepackage{graphicx}        
\usepackage{setspace}
\usepackage{tikz}
\usetikzlibrary{arrows}
\usepackage{cleveref}
\usepackage{subfigure}

\usepackage{etoolbox}
\makeatletter
\patchcmd{\@makecaption}
{\scshape}
{}
{}
{}
\makeatother

\DeclareMathAlphabet{\mathpzc}{OT1}{pzc}{m}{it}

\usepackage{epstopdf}
\usepackage{soul}
\setstcolor{red} 
\setul{0pt}{0.7pt} 

\newtheorem{defn} {Definition}

\IEEEoverridecommandlockouts

\usepackage{enumitem}
\setlist{nolistsep}

\begin{document}
\title{Edge-Based Anisotropic Decoding for Generalized Bicycle Codes}
\author{Dimitris~Chytas, Paul N. Fessatidis, Boulat A. Bash,~\IEEEmembership{Member,~IEEE,} and Bane~Vasi\'{c},~\IEEEmembership{Fellow,~IEEE}\\
\IEEEauthorblockA{Department of Electrical and Computer Engineering, University of Arizona, Tucson, AZ, USA}
Email: \{dchytas, pfessatidis, boulat\}@arizona.edu, vasic@ece.arizona.edu
\thanks{The authors acknowledge the support of the National Science Foundation under grants ERC-1941583, CIF-2420424, CIF-2106189, CCF-2100013,
CCSS-2052751, and CoQREATE program under grant ERC-1941583. We also acknowledge a generous gift from our friends and Maecenases Dora
and Barry Bursey.}
}

\maketitle

\begin{abstract}
    Quantum low-density parity-check (QLDPC) codes provide non vanishing rates, distance scaling with the blocklength of the code, and  facilitate fast iterative decoding  because of their sparsity. However, in practice iterative decoding fails to exploit the distance of the code, because it cannot resolve the symmetries imposed by degeneracy. In this work, we provide a graph theoretic characterization of degeneracy for the family of generalized bicycle (GB) codes. 
    This viewpoint shows that harmful degenerate error patterns persist whenever they remain related by automorphisms preserved by the decoder. Motivated by symmetry breaking via graph coloring, we compare three coloring approaches: no coloring, block-coloring, and edge-coloring. For GB codes, we show that edge-coloring  can eliminate all automorphisms in low-weight stabilizer-induced subgraphs. We practically realize the coloring schemes as isotropic, block-anisotropic and edge-anisotropic min-sum (MS) decoding. Experimental results show that edge anisotropic min-sum decoding  obtains improved performance over isotropic and block anisotropic decoding for several GB codes in a small number of iterations.
\end{abstract}
\section{Introduction}

\IEEEPARstart{Q}{uantum} low-density parity-check (QLDPC) codes provide superior distance and non-vanishing rates compared to topological codes~\cite{panteleev2021quantumLinearMinDLocalTestable, QuantumTannerCodeszemor}. In addition, because of their sparsity, QLDPC codes facilitate linear complexity iterative  decoding.
However, realizing constant overhead fault tolerant quantum
computation with QLDPC codes faces several remaining
challenges. One of them is designing decoders  which achieve performance that improves with increasing code
distance, while also running faster than the syndrome-extraction
process so that decoding does not become the bottleneck of
the error correction cycle~\cite{Preskill_2025,delfosse2023choosedecoderfaulttolerantquantum}.

One main property that separates quantum from classical LDPC decoding is that of degeneracy, which stems from the orthogonality constraint. Degeneracy allows the induced error pattern to be matched up to a stabilizer, which means that there are more than one equivalent error patterns that match the same syndrome~\cite{quantumTS,Fuentes_DegeneracyImpact_IEEEAccess21}. This is harmful to iterative decoding, which, because of its symmetric rules, cannot distinguish between  equivalent error patterns,  and becomes a bottleneck  when the minimum
distance increases relative to stabilizer weight\cite{osd}. To mitigate degeneracy,
a widely used approach  is iterative decoding, like belief-propagation (BP) or min-sum (MS) followed by ordered statistics, i.e., BP-OSD~\cite{osd,amb,local} which is of polynomial complexity. Other techniques achieve competitive performance either at the cost of high complexity~\cite{decim,window} or induced latency~\cite{valentini2025restartbeliefgeneralquantum,muller2025improvedbeliefpropagationsufficient}.
Approaches of lower complexity that take into account the structure of the code are proposed in~\cite{chytas2025enhancedminsumdecodingquantum,chytasISIT, pradhan2025lineartimeiterativedecoders}, with the former proposing an anisotropic decoding for bivariate bicycle (BB) codes, and the latter devising decoding techniques based on the structure of hypergraph-product and lifted-product codes.

We focus on generalized bicycle (GB) codes because they combine strong finite length performance with implementation friendly structure, and their BB subclass is a competitive candidate for low-overhead fault tolerant memory~\cite{bravyi2024high}. GB  codes are also attractive  because their algebraic construction induces a natural partition of Tanner graph edges through the monomials defining the parity-check matrices. This makes them a particularly suitable setting for the present work: the same code structure that supports implementation also provides the edge classes needed to realize fine-grained anisotropic decoding. 
Motivated by the competitive performance of the anisotropic MS decoder proposed in~\cite{chytas2025enhancedminsumdecodingquantum}, we instead focus on analyzing how anisotropy can mitigate degeneracy using graph-theoretic tools. Hence, instead of focusing on the decoding dynamics of a specific iterative decoder, we provide a graph-theoretic framework that applies to a broader class of iterative decoders. 
In particular the degeneracy problem can be described by graph automorphisms, where breaking symmetry 
can be achieved by graph coloring. 
We investigate three  coloring approaches, namely, no coloring, block-coloring, and edge-coloring which are assigned (due to equivariance) to three types of decoders, that is isotropic decoding, block-anisotropic decoding and edge-anisotropic decoding. 
In practice we realize the coloring methods using anisotropic variations of MS decoding and simulate its performance for various 6-limited and 8-limited QLDPC codes proposed in~\cite{bravyi2024high,osd,wang2022distanceboundsgeneralizedbicycle} .

The rest of the paper is organized as follows. Section~\ref{sec:pre} briefly introduces the essential preliminaries of quantum error correction (QEC). Section~\ref{sec:Theory} provides a theoretical analysis on degeneracy based on graph automorphisms and graph coloring. Finally Section~\ref{sec:practice} outlines experimental results obtained by realizing the graph coloring schemes for various codes. 

\section{Preliminaries}
\label{sec:pre}
We consider Calderbank-Shor-Steane (CSS) stabilizer codes, where $X$ and $Z$ errors are decoded separately using two binary parity-check matrices $H_Z$ and $H_X$, respectively~\cite{calderbank1996quantum_exists,Gottesman97}. In this work, we focus on decoding one error type only; without loss of generality, we consider $X$ errors and use $H:=H_Z$. Let $\mathbf{e}\in\mathbb{F}_2^N$ denote the binary error vector and $\mathbf{s}=\mathbf{e}H^T$ the measured syndrome. Decoding succeeds if $\mathbf{e}\oplus \hat{\mathbf{e}}$ belongs to the rowspace of $H$.

We focus on GB codes~\cite{osd,mackay_quantum}, a CSS family defined by two commuting binary matrices $A$ and $B$ through
\begin{equation}
H_X=[A\;B], \qquad H_Z=[B^T\;A^T],
\end{equation}
with $AB=BA$. In the algebraic description of GB codes, the matrices $A$ and $B$ are specified as sums of shift terms. Throughout the paper, we refer to each such individual shift term as a \emph{monomial}. For standard GB codes based on circulants, these monomials are single powers $x^a$ of a cyclic shift operator. An important subclass is given by BB codes~\cite{bravyi2024high}, where two commuting shifts
\begin{equation}
x=S_{L_x}\otimes I_{L_y},\qquad y=I_{L_x}\otimes S_{L_y},
\end{equation}
satisfy $xy=yx$, $x^{L_x}=y^{L_y}=I$, and the matrices $A$ and $B$ are sparse bivariate polynomials in $x$ and $y$. In this case, the monomials are terms of the form $x^a y^b$.

Writing
$
A=\sum_i A_i,~B=\sum_j B_j,$
with each $A_i,B_j$ a monomial term, each nonzero monomial naturally induces a class of edges in the Tanner graph. These monomial edge classes will serve as the labels used in our edge-anisotropic decoding framework.

The Tanner graph $G$ of $H$ is a bipartite graph with $N$ variable nodes $v\in\{v_1,\dots,v_N\}$ and $M$ check nodes $c\in\{c_1,\dots,c_M\}$, where $(v,c)$ is an edge whenever $H_{vc}=1$. We denote by $\mathcal{M}(v)$ the set of checks neighboring variable node $v$, and by $\mathcal{N}(c)$ the set of variables neighboring check node $c$. The length of the shortest cycle in $G$ is referred to as girth.

We define two classes of syndrome-based decoding algorithms that we will use in the following sections.
\begin{defn}
\label{defn:isotropic}
Let $G=(\mathcal V,\mathcal C,\mathcal E)$ be a Tanner graph associated with a
parity-check matrix $H$ and a syndrome vector $\boldsymbol{s}$. Define as $c \in \mathcal{C}$ and $v \in \mathcal{V}$, a check and a variable node of $G$, and as $\mathcal{N}(c)$, $\mathcal{M}(v)$ their neighborhoods.
An \emph{isotropic syndrome-based decoder} is an iterative message-passing decoder
defined by update functions $\Phi$ and $\Psi$ such that, for all iterations $\ell\ge1$,
the messages evolve according to the flooding schedule
\begin{align}
\nu^{(\ell)}_{v\to c}
&=
\Phi\!\left(
\lambda,\,
\{\mu^{(\ell-1)}_{c'\to v} : c' \in \mathcal{M}(v)\setminus\{c\}\}
\right),
\label{eq:isotropic_v}
\\
\mu^{(\ell)}_{c\to v}
&=
\Psi\!\left(
s_c,\,
\{\nu^{(\ell)}_{v'\to c} : v' \in \mathcal{N}(c)\setminus\{v\}\}
\right),
\label{eq:isotropic_c}
\\
\nu^{(\ell)}_{v}
&=
\Phi\!\left(
\lambda,\,
\{\mu^{(\ell)}_{c\to v} : i \in \mathcal{M}(v)\}
\right),
\label{eq:isotropic_dec}
\end{align}
where $\nu^{(\ell)}_{v\to c}$ and $\mu^{(\ell)}_{c\to v}$ denote the variable-to-check and
check-to-variable messages, respectively, and $\nu^{(\ell)}_{v}$ denotes the tentative
decision at variable node $v \in \mathcal V$. The scalar $\lambda$ denotes a \emph{uniform prior} shared by all
variable nodes.

The decoder is called \emph{isotropic} if: i) the same update functions $\Phi$ and $\Psi$ are used at all variable and check nodes respectively, ii)  $\Phi$ and $\Psi$ depend only on the \emph{multiset} of incoming messages and are invariant under permutations of their arguments, and iii)  $\Phi$ and $\Psi$ are independent of the iteration index $\ell$.

\end{defn}

\begin{defn}
\label{defn:anisotropic}
Let $G=(\mathcal V,\mathcal C,\mathcal E)$ be a Tanner graph with syndrome labels $\boldsymbol{s}$. Suppose a family of variable-node update rules
\[
\Phi^{(1)},\dots,\Phi^{(K)}
\]
is given, while the check-node update rule $\Psi$ is shared across all check nodes.

A decoder is called \emph{variable-anisotropic} if there exists a labeling
\[
\chi_V:\mathcal V\to\{1,\dots,K\}
\]
such that, for each variable node $v\in\mathcal V$, all outgoing messages from $v$ are updated using the same rule $\Phi^{(\chi_V(v))}$.

A decoder is called \emph{edge-anisotropic} if there exists a labeling
\[
\chi_E:\mathcal E\to\{1,\dots,K\}
\]
such that, for each edge $(v,c)\in\mathcal E$, the message $\nu^{(\ell)}_{v\to c}$ is updated using the rule $\Phi^{(\chi_E(v,c))}$. In this case, the same variable node may use different update rules on different outgoing edges.

If all labels induce the same update rule, the decoder reduces to an isotropic syndrome-based decoder.
\end{defn}
We note that for the rest of the paper the terms coloring and labeling will be  used interchangeably.

\section{Symmetry Framework for Degeneracy}
\label{sec:Theory}
In this Section we characterize harmful degenerate errors for different types of iterative decoders using graph automorphisms.

\subsection{Decoder Equivariance and Symmetry Preservation}

For the following, the isolation assumption~\cite{iso} is considered, for which all nodes outside
the induced subgraph of interest are considered to have converged and can be excluded during the check-to-variable updates of the decoder.

Before we proceed, we provide  definitions of graph automorphisms that will be useful throughout the paper.
\begin{defn}
\label{def:labeled_aut}
Let $G=(\mathcal V,\mathcal C,\mathcal E)$ be an induced Tanner subgraph. A \emph{type-preserving automorphism} of $G$ is a bijection
\[
\psi:\mathcal V\sqcup\mathcal C \to \mathcal V\sqcup\mathcal C
\]
such that
\begin{enumerate}
    \item $\psi(\mathcal V)=\mathcal V$ and $\psi(\mathcal C)=\mathcal C$,
    \item $(v,c)\in\mathcal E$ if and only if $(\psi(v),\psi(c))\in\mathcal E$.
\end{enumerate}
The set of all such automorphisms is denoted by $\operatorname{Aut}(G)$.

Suppose in addition that $G$ is equipped with one or more labelings:
\[
\chi_V:\mathcal V\to X_V,\qquad
\chi_C:\mathcal C\to X_C,\qquad
\chi_E:\mathcal E\to X_E,
\]
where $X_V,X_C,X_E$ are label sets. Here $\chi_C$ may, for example, represent the syndrome labeling on check nodes, while $\chi_V$ or $\chi_E$ may represent decoder-dependent labelings.

A type-preserving automorphism $\psi\in\operatorname{Aut}(G)$ is said to \emph{preserve} a labeling if it preserves its value on the corresponding domain, namely
\[
\chi_V(v)=\chi_V(\psi(v)) \quad \forall v\in\mathcal V,
\]
\[
\chi_C(c)=\chi_C(\psi(c)) \quad \forall c\in\mathcal C,
\]
and
\[
\chi_E(e)=\chi_E(\psi(e)) \quad \forall e\in\mathcal E.
\]
For any specified collection of labelings, we write $\operatorname{Aut}(G,\chi)$ for the subgroup of $\operatorname{Aut}(G)$ preserving all labelings in that collection.
\end{defn}

\begin{lemma}
\label{lem:decoder_equivariance}
Let $G=(\mathcal V,\mathcal C,\mathcal E)$ be an induced Tanner subgraph with syndrome vector
$\boldsymbol{s}$, viewed as a check-node labeling. Consider a flooding-schedule
syndrome-based iterative decoder on $G$, and assume the isolation assumption holds.

Let $\psi\in \operatorname{Aut}(G,\chi_C)$, where $\chi_C(c)=s_c$ for all $c\in\mathcal C$.
Assume moreover that the initial variable-to-check messages and any boundary messages are
invariant under $\psi$.

If the decoder is isotropic, then for every iteration $\ell\ge 0$ and every edge $(v,c)\in\mathcal E$,
\[
\nu^{(\ell)}_{v\to c}=\nu^{(\ell)}_{\psi(v)\to\psi(c)},
\qquad
\mu^{(\ell)}_{c\to v}=\mu^{(\ell)}_{\psi(c)\to\psi(v)},
\]
and for every variable node $v\in\mathcal V$,
\[
\nu^{(\ell)}_{v}=\nu^{(\ell)}_{\psi(v)}.
\]

If the decoder is anisotropic with decoder labeling $\chi$, defined either on $\mathcal V$ or on
$\mathcal E$, and if in addition $\psi\in\operatorname{Aut}(G,\chi_C,\chi)$, then the same conclusion holds.

Hence the message evolution and tentative decisions are equivariant under every automorphism
preserved by the decoder.
\end{lemma}
\begin{proof}
We prove the claim by induction on the iteration index $\ell$.

At initialization, the variable-to-check messages and all boundary messages are invariant
under $\psi$ by assumption, so the claim holds for $\ell=0$.

Now assume that at iteration $\ell$,
\[
\nu^{(\ell)}_{v\to c}=\nu^{(\ell)}_{\psi(v)\to\psi(c)},
\qquad
\mu^{(\ell)}_{c\to v}=\mu^{(\ell)}_{\psi(c)\to\psi(v)}
\]
for every edge $(v,c)\in\mathcal E$, and similarly that tentative decisions agree on
automorphism-related variable nodes.

Since $\psi$ preserves adjacency, it induces bijections
\[
\mathcal M(v)\setminus\{c\}\longrightarrow \mathcal M(\psi(v))\setminus\{\psi(c)\}
\]
and
\[
\mathcal N(c)\setminus\{v\}\longrightarrow \mathcal N(\psi(c))\setminus\{\psi(v)\}.
\]
By the induction hypothesis, corresponding incoming messages in these neighborhoods are equal.
Moreover, $\psi$ preserves all local data used by the update rules: syndrome labels in the isotropic
case, and additionally the decoder labeling $\chi$ in the anisotropic case. Therefore corresponding
nodes or edges apply the same local update rule to the same multiset of incoming messages. Since the
local update maps are permutation invariant in their incoming arguments, it follows that
\[
\nu^{(\ell+1)}_{v\to c}=\nu^{(\ell+1)}_{\psi(v)\to\psi(c)},
\qquad
\mu^{(\ell+1)}_{c\to v}=\mu^{(\ell+1)}_{\psi(c)\to\psi(v)}.
\]

The same argument applies to the tentative decision update at variable nodes, yielding
\[
\nu^{(\ell+1)}_{v}=\nu^{(\ell+1)}_{\psi(v)}.
\]
\end{proof}

For the following analysis, we assume that a stabilizer generator of a GB code induces a Tanner subgraph isomorphic to the subdivided complete bipartite graph $K_{m,n}$. For an error pattern $\mathcal E$ supported on the variable-node support of such a stabilizer, we define its complementary pattern
\[
\mathcal F:=\mathcal V\setminus\mathcal E.
\]
Since $\mathcal F$ differs from $\mathcal E$ by the full stabilizer support, the two patterns are degenerate. Moreover, if $|\mathcal E|=|\mathcal F|$, then necessarily
\[
|\mathcal E|=|\mathcal F|=\frac{m+n}{2},
\]
which requires $m+n$ to be even. Finally, $\mathcal E$ and $\mathcal F$ induce identical syndrome labels on the degree-$2$ check nodes, because a check is unsatisfied if and only if exactly one of its two neighbors belongs to the error pattern, and this condition is unchanged by complementation. The next theorem characterizes when these two degenerate patterns remain related by graph automorphisms under no coloring and under block coloring.

\begin{theorem}
\label{thm:none_block_subdivided_Kmn}
Let $T$ be a stabilizer-induced Tanner subgraph isomorphic to the subdivided complete bipartite graph $K_{m,n}$, with variable-node partition
\[
\mathcal V=\mathcal A\sqcup\mathcal B,
\qquad |\mathcal A|=m,\quad |\mathcal B|=n.
\]
For an error pattern $\mathcal E\subseteq\mathcal V$, define
\[
\mathcal F:=\mathcal V\setminus\mathcal E,
\qquad
x:=|\mathcal E\cap\mathcal A|,
\qquad
y:=|\mathcal E\cap\mathcal B|.
\]

Then the following hold.

\begin{enumerate}
    \item Under no coloring (equivalently, for isotropic decoding), there exists a syndrome-label-preserving automorphism
    \[
    \psi\in\mathrm{Aut}(T)
    \qquad\text{such that}\qquad
    \psi(\mathcal E)=\mathcal F
    \]
    if and only if either
    \begin{enumerate}
        \item $m\neq n$ and
        \[
        x=m-x,\qquad y=n-y,
        \]
        or
        \item $m=n$ and $|\mathcal E|=m$.
    \end{enumerate}

    \item Under block coloring $C_{\mathrm{blk}}$ (equivalently, for block-anisotropic decoding), there exists a syndrome- and color-preserving automorphism
    \[
    \psi\in\mathrm{Aut}(T,C_{\mathrm{blk}})
    \qquad\text{such that}\qquad
    \psi(\mathcal E)=\mathcal F
    \]
    if and only if
    \[
    x=m-x,\qquad y=n-y.
    \]
\end{enumerate}

Hence block coloring changes the failure condition only in the symmetric case $m=n$, where it removes the side-swapping automorphisms present under no coloring.
\end{theorem}
\begin{proof}
Let
\[
x:=|\mathcal E\cap\mathcal A|,
\qquad
y:=|\mathcal E\cap\mathcal B|.
\]
Then
\[
|\mathcal E|=x+y,
\qquad
|\mathcal F|=(m-x)+(n-y)=m+n-(x+y).
\]

Since the full stabilizer support is precisely $\mathcal V=\mathcal A\sqcup\mathcal B$, the pattern $\mathcal F=\mathcal V\setminus\mathcal E$ differs from $\mathcal E$ by that stabilizer, and therefore $\mathcal E$ and $\mathcal F$ are degenerate. Moreover,
\[
|\mathcal E|=|\mathcal F|
\quad\Longleftrightarrow\quad
x+y=m+n-(x+y),
\]
which is equivalent to
\[
|\mathcal E|=|\mathcal F|=\frac{m+n}{2}.
\]

Next, consider any degree-$2$ check node of $T$. It is adjacent to one variable in $\mathcal A$ and one in $\mathcal B$, and its syndrome label is $1$ if and only if exactly one of its two neighbors belongs to the error pattern. This condition is unchanged when $\mathcal E$ is replaced by its complement $\mathcal F$, so $\mathcal E$ and $\mathcal F$ induce identical syndrome labels on the checks. 

Suppose first that $m\neq n$. The variables in $\mathcal A$ have degree $n$, while those in $\mathcal B$ have degree $m$. Since graph automorphisms preserve degree, every automorphism of $T$ preserves $\mathcal A$ and $\mathcal B$ setwise. Moreover, because $T$ is the subdivided $K_{m,n}$, any permutation within $\mathcal A$ together with any permutation within $\mathcal B$ extends uniquely to an automorphism of $T$. Hence, under either no coloring or block coloring, an automorphism can map $\mathcal E$ to $\mathcal F$ if and only if $\mathcal E$ and $\mathcal F$ select the same number of variables from each side, namely
\[
x=m-x,
\qquad
y=n-y.
\]
Equivalently,
\[
x=\frac m2,
\qquad
y=\frac n2.
\]

Now assume $m=n$. In the no-coloring case, the subdivided complete bipartite graph admits, in addition to independent permutations within the two sides, an automorphism that swaps $\mathcal A$ and $\mathcal B$. If $|\mathcal E|=m$, then $x+y=m$ and the complement $\mathcal F$ has type
\[
(|\mathcal F\cap\mathcal A|,|\mathcal F\cap\mathcal B|)=(m-x,m-y)=(y,x).
\]
Thus a side-swapping automorphism maps a pattern of type $(x,y)$ to one of type $(y,x)$, so every weight-$m$ pattern is related to its complement. 

Under block coloring, however, side-swapping is forbidden  when $m=n$, since $\mathcal A$ and $\mathcal B$ carry different colors. Therefore the only remaining color-preserving automorphisms are permutations within each side, and $\mathcal E$ can be mapped to $\mathcal F$ if and only if
\[
x=m-x,
\qquad
y=n-y.
\]
Hence block coloring leaves residual harmful symmetry only for patterns balanced within each block. 
\end{proof}

A more granular approach that removes automorphisms that persist for both cases of Theorem~\ref{thm:none_block_subdivided_Kmn} is presented below. The theorem applies to GB codes where each block $A,B$ is specified as a sum of distinct monomials, therefore each monomial corresponds to a set of edges. 

\begin{theorem}
\label{thm:gb_edge_rigidity}
Let $T=(\mathcal V,\mathcal C,\mathcal E)$ be a stabilizer-induced subgraph of a generalized bicycle code, and suppose that $T$ is isomorphic to a subdivided $K_{m,n}$ with variable-node partition
\[
\mathcal V=\mathcal A\sqcup\mathcal B,
\qquad |\mathcal A|=m,\quad |\mathcal B|=n.
\]
Assume that the edges incident to variables nodes of $\mathcal A$ are labeled by $m$ distinct edge labels, each corresponding to each of the monomials:
\[
\alpha_1,\dots,\alpha_m,
\]
while the edges incident to variables nodes of $\mathcal B$ are labeled by $n$ distinct edge labels, each corresponding to each of the monomials:
\[
\beta_1,\dots,\beta_n,
\]
with
\[
\{\alpha_1,\dots,\alpha_m\}\cap \{\beta_1,\dots,\beta_n\}=\varnothing.
\]
Then every degree-$2$ check node $c\in\mathcal C$ is uniquely identified by its ordered  pair
\[
\eta(c)=\bigl(\alpha_i,\beta_j\bigr),
\qquad 1\le i\le m,\; 1\le j\le n,
\]
and all such ordered pairs appear exactly once. Consequently,
no nontrivial color-preserving automorphism can map one error pattern to a distinct error pattern.
\end{theorem}
\begin{proof}
Since the edge labels on the $\mathcal A$ side belong to the family
$\{\alpha_1,\dots,\alpha_m\}$ and those on the $\mathcal B$ side belong to the disjoint family
$\{\beta_1,\dots,\beta_n\}$, any color-preserving automorphism must preserve
$\mathcal A$ and $\mathcal B$ setwise. Because $T$ is isomorphic to a subdivided $K_{m,n}$,
each degree-$2$ check joins exactly one variable on the $\mathcal A$ side to one variable on the
$\mathcal B$ side. Hence every check node is associated with a unique ordered pair
\[
(\alpha_i,\beta_j),
\qquad 1\le i\le m,\; 1\le j\le n.
\]
Since there are exactly $mn$ such ordered pairs and exactly $mn$ degree-$2$ checks in a subdivided
$K_{m,n}$, each ordered pair appears exactly once.

Let $\psi\in\mathrm{Aut}(T,\chi)$ be color-preserving. Because $\psi$ preserves the two sides and
the edge labels, it preserves the ordered incident monomial pair of every check node. By uniqueness
of these pairs, every degree-$2$ check node is fixed. Each variable node is then uniquely determined
by its adjacent degree-$2$ checks, so all variable nodes are fixed as well. Therefore
$\psi$ is the identity.
\end{proof}

\subsection{Linear Combinations of Stabilizer Generators}
In this section, we focus on the $H_Z$ parity check matrix of $6$-limited and $8$-limited codes that have girth equal to six. For $6$-limited codes, we pick two representatives, one is the $[[288,12,18]]$ BB code\cite{bravyi2024high}, and the other one is the  blocklength $166$ and distance $18$ code from~\cite{wang2022distanceboundsgeneralizedbicycle}. The  induced subgraph of the stabilizer generator of each code is isomorphic to the subdivided $K_{3,3}$ and $K_{2,4}$ complete bipartite graphs respectively. Regarding $8$-limited codes, we choose the $[[180,10,15\leq d\leq 18]]$ $A5$ code~\cite{osd}, for which the induced subgraph of its stabilizer generator is isomorphic to the subdivided $K_{4,4}$  complete bipartite graph. We count error patterns  supported on low-weight stabilizers, for which there exists an automorphism mapping the pattern to a degenerate error pattern of the same weight. We perform this analysis under three labeling schemes: no coloring, block-coloring and edge-coloring. Stabilizers are obtained by linear combinations of  stabilizer generators and here we  restrict attention to the three minimum weight cases: a single stabilizer generator; two stabilizer generators sharing one variable node; and three stabilizer generators forming a $6$-cycle in the $H_X$ matrix, so that each pair shares exactly one variable node. As an example, Fig.~\ref{fig:stabilizers} illustrates stabilizers formed by stabilizer generators that share one variable node. 

Enumeration is carried out using nauty~\cite{mckay2013practicalgraphisomorphismii}, and the obtained results are shown in Table~\ref{tab:orbit_counts}. In the table, \emph{Single} denotes a single stabilizer generator. The entries \emph{shared A} and \emph{shared B} denote pairs of generators sharing one variable node in block $A$ or block $B$, respectively. The entries $(A,A,A)$ and $(A,A,B)$ denote triples of generators forming a $6$-cycle, where the shared variable nodes lie either all in block $A$ or in blocks $A,A,B$. For the blocklength-$166$ code, whose stabilizer generators induce subdivided $K_{2,4}$ graphs, the two sides are not exchangeable since $m\neq n$, so configurations involving shared nodes in different blocks must be distinguished. In particular, no $6$-cycles involving more than one shared $B$-side node occur in this code. By contrast, for the $K_{3,3}$ and $K_{4,4}$ cases we have $m=n$, so the two sides are exchangeable and configurations such as $(A,B,B)$ or $(B,B,B)$ are isomorphic to the representatives already listed.
\begin{figure}
    \centering
    \begin{subfigure}[]
    {
        \centering
        \includegraphics[width=.2\textwidth]{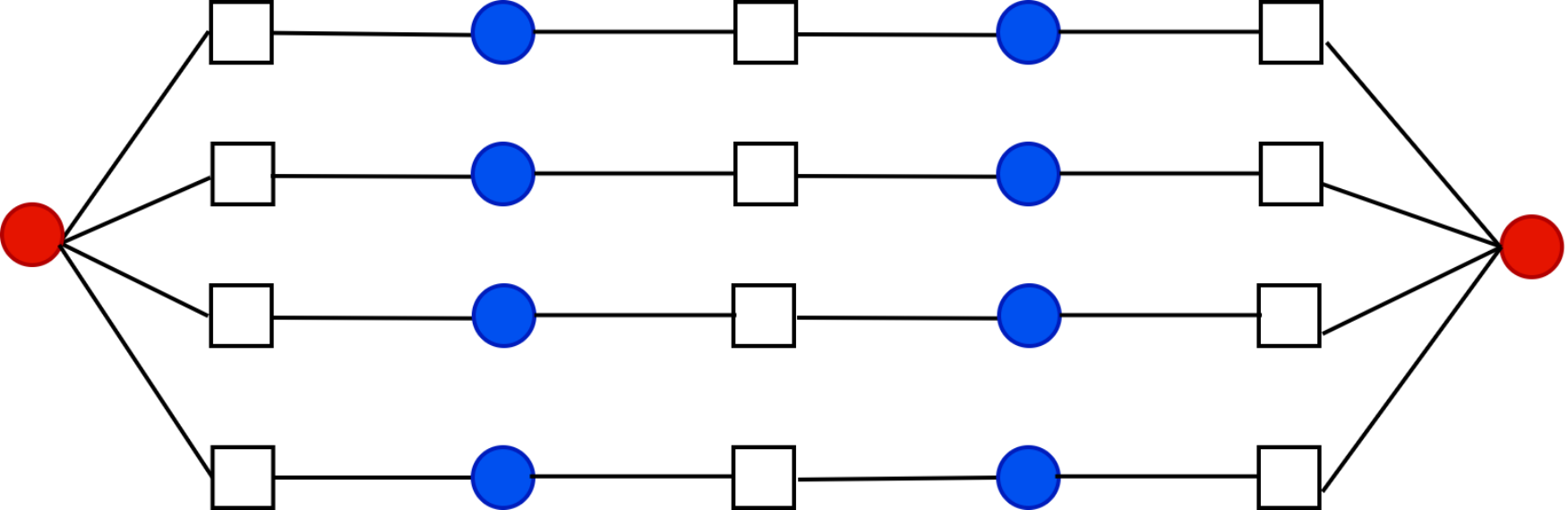}
        \label{fig:stab24a}
        }
    \end{subfigure}
    \begin{subfigure}[]
    {
        \centering
        \includegraphics[width=.13\textwidth]{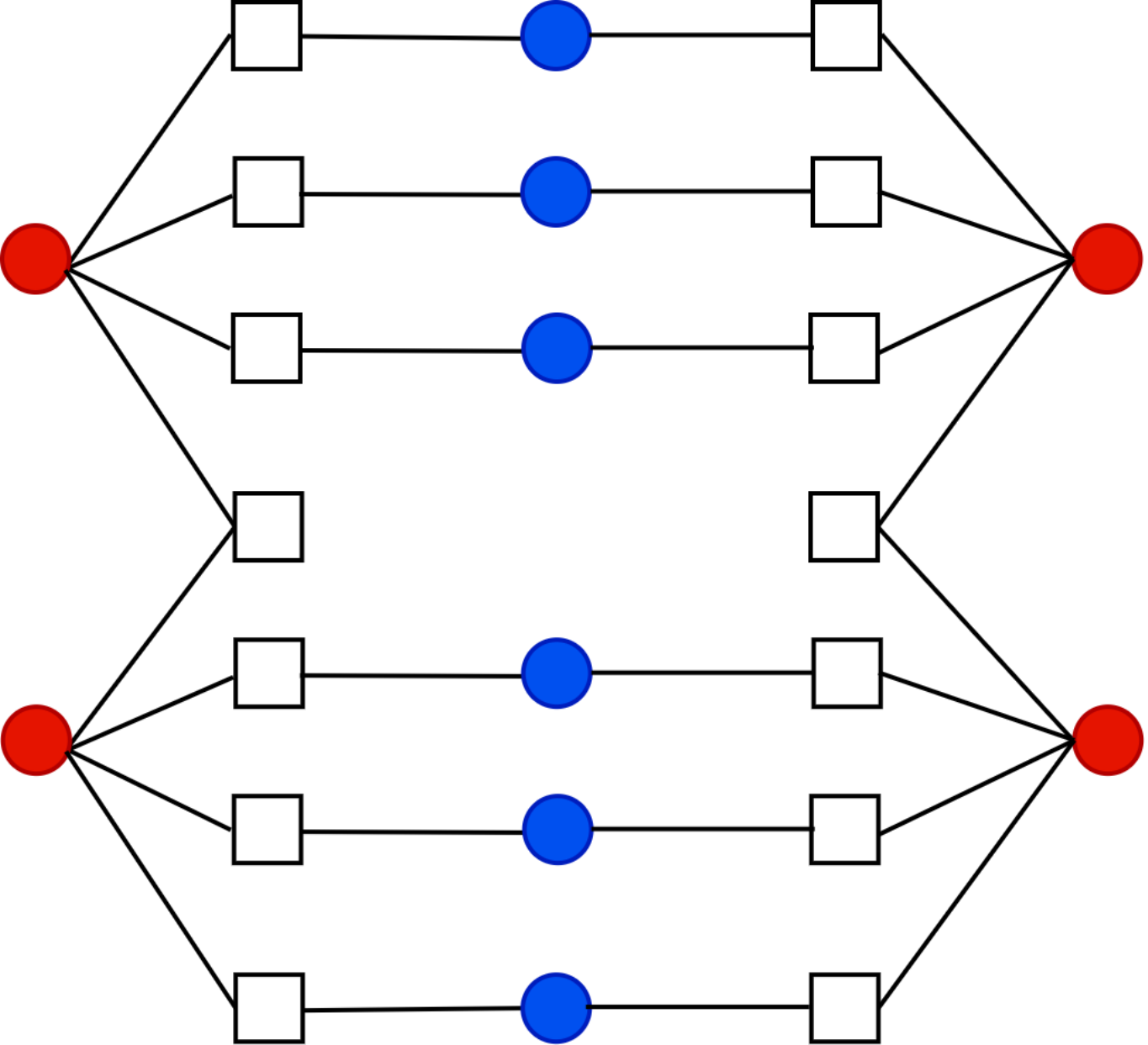}
        \label{fig:stab24b}
        }
    \end{subfigure}
    \begin{subfigure}[]
    {
        \centering
        \includegraphics[width=.3\textwidth]{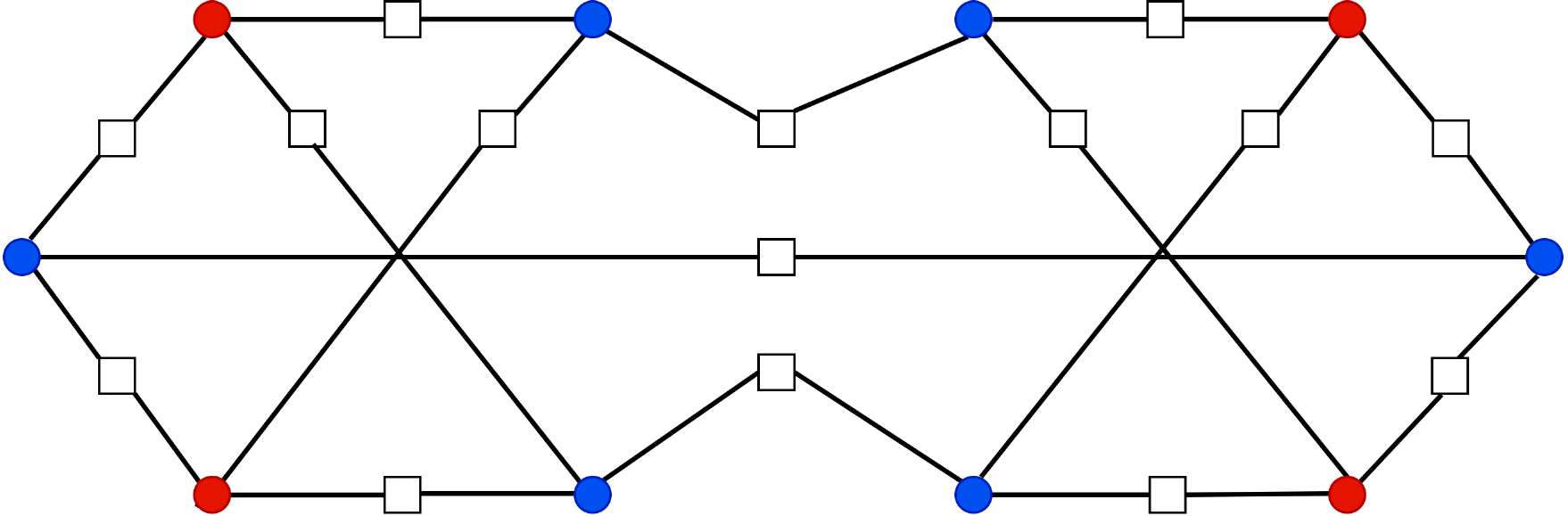}
        \label{fig:stab33}
        }
     \end{subfigure}
    \caption{Induced subgraphs of stabilizers formed by linear combinations of stabilizer generators isomorphic to the subdivided $K_{2,4}$ \subref{fig:stab24a},\subref{fig:stab24b} and to the subdivided $K_{3,3}$ \subref{fig:stab33} complete bipartite graph. Variable nodes are color coded based on the block they belong to.} 
    \label{fig:stabilizers}
\end{figure}

\begin{table}[t]
\centering
\caption{Number of harmful configurations mapped by a graph automorphism to a degenerate error pattern of the same weight, for families of $6$-limited and $8$-limited GB codes.}
\label{tab:orbit_counts}
\footnotesize
\setlength{\tabcolsep}{3pt}
\renewcommand{\arraystretch}{1.08}
\resizebox{\columnwidth}{!}{%
\begin{tabular}{llccccc}
\toprule
\textbf{Stab.} & \textbf{Coloring} & \textbf{Single}
& \shortstack{\textbf{2-gen}\\ \textbf{shared $A$}}
& \shortstack{\textbf{2-gen}\\ \textbf{shared $B$}}
& \shortstack{\textbf{6-cycle}\\ \textbf{$(A,A,A)$}}
& \shortstack{\textbf{6-cycle}\\ \textbf{$(A,A,B)$}} \\
\midrule
\multirow{3}{*}{$K_{3,3}$}
& None          & $10$ & $46$  & $46$ & $0$ & $0$  \\
& Block         & $0$ & $46$  & $46$ & $0$ & $0$ \\
& Edge/monomial & $0$ &  $0$& $0$ & $0$ & $0$ \\
\midrule
\multirow{3}{*}{$K_{2,4}$}
& None          & $6$ & $52$ & $46$ & $44$ & $28$  \\
& Block         & $6$ & $52$ & $46$ & $44$ & $28$ \\
& Edge/monomial & $0$ & $0$ & $0$ & $0$ & $0$ \\
\midrule
\multirow{3}{*}{$K_{4,4}$}
& None          & $35$ & $87$  & $87$ & $0$ &$0$  \\
& Block         & $18$ & $87$ & $87$ & $0$ & $0$ \\
& Edge/monomial & $0$ & $0$ & $0$ & $0$ & $0$ \\
\bottomrule
\end{tabular}%
}
\end{table}

\section{Experimental results}
\label{sec:practice}
MS decoding is a low-complexity approximation of BP~\cite{05CDEFH}.  
Since no channel information is available, we use the syndrome-based MS variant. A check-to-variable message sent from check node $c$ to variable node $v$ at iteration $\ell$ is denoted by $\mu^{(\ell)}_{c,v}$, whereas a variable-to-check message sent from variable node $v$ to check node $c$ is denoted by $\nu^{(\ell)}_{v,c}$ and computed as
\begin{equation}
    \nu^{(\ell)}_{v,c} = \lambda + \sum\limits_{c' \in \mathcal{M}(v) \backslash \{ c \}}\mu^{(\ell)}_{c',v},
    \label{eq:vnu}
\end{equation}
where $\lambda=\log\!\left(\frac{1-\alpha}{\alpha}\right)$ is the a priori log-likelihood ratio (LLR) of every variable node. MS complies with Definition~\ref{defn:isotropic}, so we refer to it as \emph{isotropic MS} decoder.

A deterministic anisotropic MS rule introduced in~\cite{chytas2025enhancedminsumdecodingquantum} applies oscillation damping on a set of variable nodes only after checking whether the sign of the pre-computed variable-to-check message disagrees with that of the previous message on the same edge. In this work, however, we adopt a simpler realization that is more closely aligned with the  framework of Section~\ref{sec:Theory}. Specifically, we realize anisotropy through label-dependent damping applied directly in the variable-to-check update at every iteration. Thus, the edge labeling is fixed throughout decoding, and the label-dependent part of the update is always active. This yields a direct realization of graph coloring: different edge classes apply different local rules at all times. By contrast the conditional sign-flip approach is a dynamics-aware variant, for which the  anisotropic part of the update is activated only on edges whose messages reverse sign. Since the focus of this paper is the short iteration regime, we use unconditional damping as the primary realization, because it introduces the intended symmetry breaking immediately and makes the effect of decoder labeling on early convergence easier to assess. 

Let $\chi:\mathcal E\to\{1,\dots,K\}$
be an edge-class labeling function. Then the variable-to-check update is defined as
\begin{equation}
    \nu^{(\ell)}_{v,c}
    =
    \xi_{\chi(v,c)}\tilde{\nu}^{(\ell)}_{v,c}
    +
    \bigl(1-\xi_{\chi(v,c)}\bigr)\nu^{(\ell-1)}_{v,c},
    \label{eq:unified_xi}
\end{equation}
where $\tilde{\nu}^{(\ell)}_{v,c}$ is the pre-computed message obtained from Eq.~\eqref{eq:vnu}. 

Eq.~\eqref{eq:unified_xi} provides a unified description of the anisotropic decoders considered in this work. In particular, the block-anisotropic decoder of~\cite{chytas2025enhancedminsumdecodingquantum} is recovered, without the conditional sign-flip mechanism, by choosing two edge classes according to each block and assigning:
\[
\xi_{\chi(v,c)}=
\begin{cases}
0.5, & \text{if } v \text{ belongs to the damped block},\\
1, & \text{otherwise}.
\end{cases}
\]
\begin{figure}[t]
    \centering
    \begin{subfigure}[]
    {
        \centering
        \includegraphics[width=.42\textwidth]{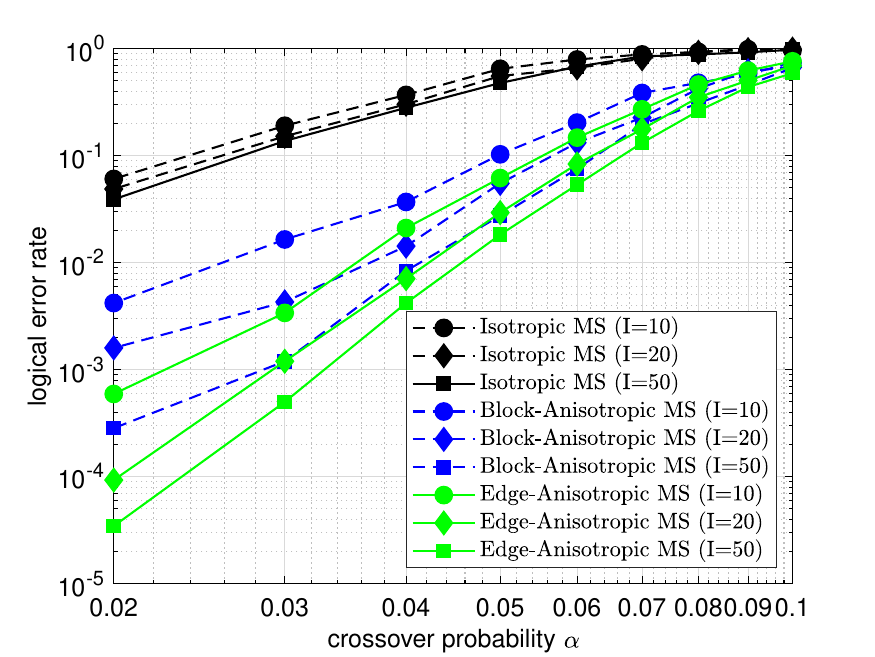}
        \label{fig:24a}
        }
    \end{subfigure}
    \begin{subfigure}[]
    {
        \centering
        \includegraphics[width=.42\textwidth]{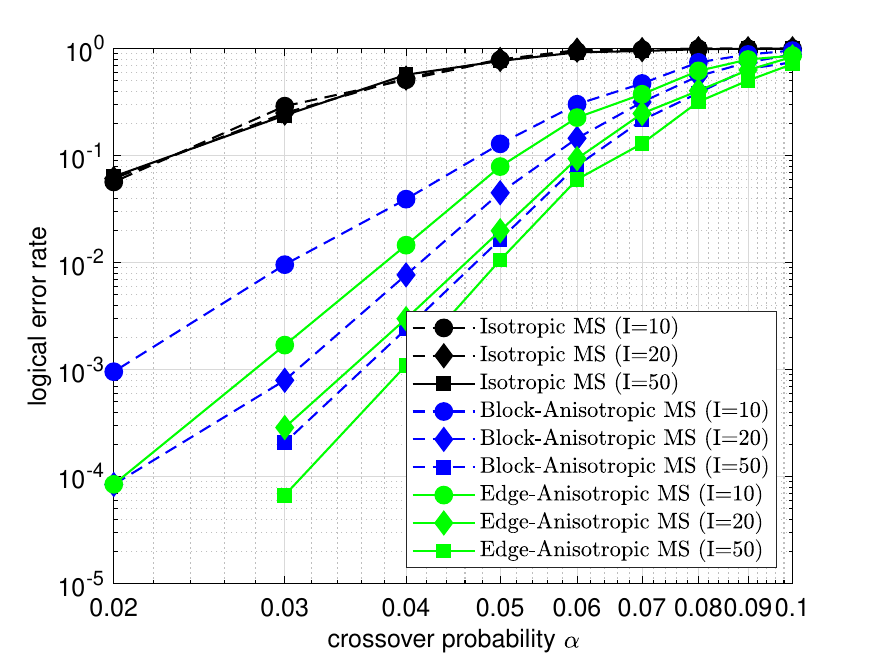}
        \label{fig:24b}
        }
    \end{subfigure}
    \begin{subfigure}[]
    {
        \centering
        \includegraphics[width=.42\textwidth]{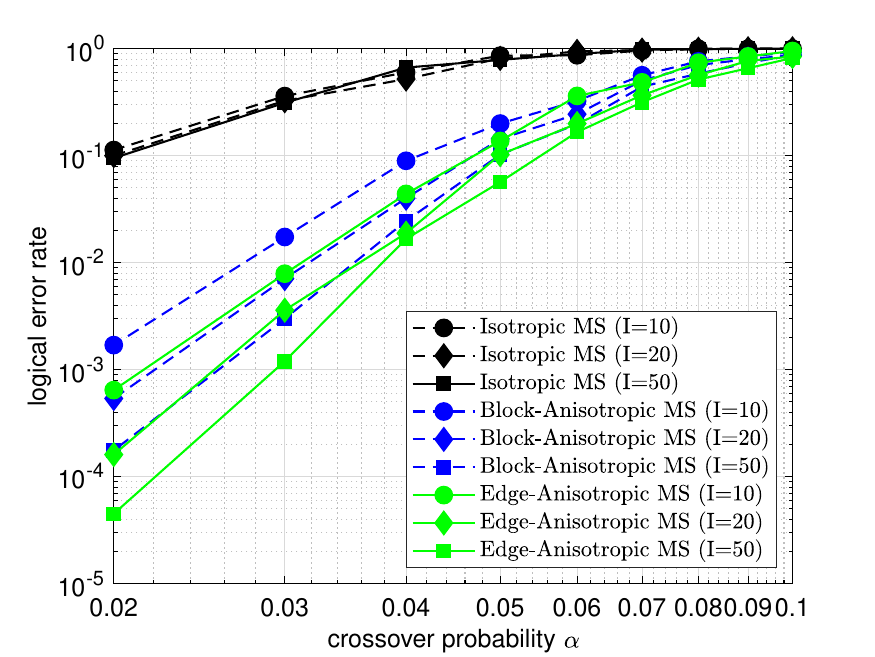}
        \label{fig:33}
        }
     \end{subfigure}
 \caption{Logical error rate under $X$ errors for isotropic MS, block-anisotropic MS, and an ensemble realization of edge-anisotropic MS on three representative GB codes: \subref{fig:24a} the GB code of~\cite{wang2022distanceboundsgeneralizedbicycle}, \subref{fig:24b} the $[[288,12,18]]$ BB code, and \subref{fig:33} the A5 code~\cite{osd}. Results are shown for iteration budgets $I\in\{10,20,50\}$.}
    \label{fig:performance}
\end{figure}

As shown in Section~\ref{sec:Theory}, block anisotropy does not remove all harmful automorphisms, which motivates a finer realization based on the monomial structure of the code. In particular, each monomial term induces a class of Tanner graph edges, and we assign $\chi(v,c)$ according to the monomial label of edge $(v,c)$. The resulting anisotropic decoder uses the same update rule as in Eq.~\eqref{eq:unified_xi}, but with one damping parameter $\xi_k$ per monomial edge class. The parameters $\xi_k$, $k\in\{1,\dots,K\}$, satisfy $0.5 \le \xi_k \le 1$. We restrict $\xi_k$ to this interval to avoid  message stalling, as we  observed degraded performance for values below $0.5$. For all experiments, we fix the normalization parameter to one in order to isolate the role of anisotropy itself.

We evaluate the $6$-limited and $8$-limited codes of Section~\ref{sec:Theory} under isotropic, block-anisotropic, and edge-anisotropic MS decoding, using iteration budgets $I\in\{10,20,50\}$. For the $6$-limited codes, each candidate is specified by six distinct damping parameters $\{\xi_k\}_{k=1}^{6}$, one per monomial edge class, while for the $8$-limited code we use eight parameters $\{\xi_k\}_{k=1}^{8}$; in all cases, the parameters satisfy $0.5\le \xi_k\le 1$. Since a principled optimal choice of monomial-level damping parameters is not yet available, and more structured methods are reserved for future work, we do the following.   We generate a random pool of $100$ candidate decoders and  evaluate it on the harmful degenerate error patterns supported on the stabilizer configurations of Table~\ref{tab:orbit_counts} by selecting five decoders that collectively give the best coverage within $10$ iterations. These five decoders are used as the ensemble realization of edge anisotropy in the performance experiments below. For the BB code, this ensemble corrects all such configurations within $10$ iterations, while for the other two codes it corrects most of them. Our emphasis on a small iteration budget is motivated both by low-latency decoding requirements and by the fact that, in this regime, decoder behavior is better approximated by the isolation assumption~\cite{iso} and therefore more directly interpretable through the analysis of Section~\ref{sec:Theory}.

In Fig.~\ref{fig:performance}, we compare isotropic MS, block-anisotropic MS, and an ensemble realization of edge-anisotropic MS for the representative GB codes of Table~\ref{tab:orbit_counts}. For all three codes, isotropic MS shows little or no improvement as the iteration budget increases. In Fig.~\ref{fig:performance}\subref{fig:24a}, block-anisotropic MS improves over isotropic MS, but its gains remain limited. This is consistent with Table~\ref{tab:orbit_counts}, where block coloring does not remove all harmful automorphisms for the corresponding stabilizer configurations. By contrast, the edge-anisotropic ensemble yields the strongest and most consistent improvement as the number of iterations increases. 
In Fig.~\ref{fig:performance}\subref{fig:24b}, block-anisotropic decoding performs more favorably, which is consistent with the fact that the relevant stabilizers are isomorphic to subdivided $K_{3,3}$ configurations, for which block coloring can already remove the dominant automorphisms. Even in this case, the monomial-edge realization remains superior across all iteration budgets. Similar behavior is observed for the A5 code in Fig.~\ref{fig:performance}\subref{fig:33}.

\section{Conclusions and future work}
Motivated by a decoder-agnostic, graph-theoretic analysis of degeneracy, we show that it can be  mitigated by employing edge-coloring. This is applicable to highly structured and practical codes of the literature such as GB codes, where coloring can be applied to sets of edges defined by each monomial. As a practical realization of this approach, we propose an edge-anisotropic MS decoder which is able to correct harmful degenerate errors in the support of low-weight linear combinations of stabilizer generators within a small number of iterations. 
Future work will be focused on the extension of such scheme for different families of codes and more realistic error models, starting with the phenomenological noise model. Another natural direction is to develop more principled decoder selection methods, possibly via machine learning, which we reserve for future work.

\IEEEtriggeratref{14} 
\bibliographystyle{IEEEtran}
\bibliography{bib/totalRefs.bib}

@misc{chytas2025enhancedminsumdecodingquantum,
      title={Enhanced Min-Sum Decoding of Quantum Codes Using Previous Iteration Dynamics}, 
      author={Dimitris Chytas and Nithin Raveendran and Bane Vasic},
      year={2025},
      eprint={2501.05021},
      archivePrefix={arXiv},
      primaryClass={quant-ph},
      url={https://arxiv.org/abs/2501.05021}, 
}

@misc{valentini2025restartbeliefgeneralquantum,
      title={Restart Belief: A General Quantum LDPC Decoder}, 
      author={Lorenzo Valentini and Diego Forlivesi and Andrea Talarico and Marco Chiani},
      year={2025},
      eprint={2511.13281},
      archivePrefix={arXiv},
      primaryClass={quant-ph},
      url={https://arxiv.org/abs/2511.13281}, 
}

@article{quantumTS,
  title = {Trapping Sets of Quantum {LDPC} Codes},
  author = {Raveendran, Nithin and Vasi{\'{c}}, Bane},
  journal = {{Quantum}},
  issn = {2521-327X},
  doi = {10.22331/q-2021-10-14-562},
  volume = {5},
  pages = {562},
  month = {Oct.},
  year = {2021}
}

@article{osd,
	doi = {10.22331/q-2021-11-22-585},

  	year = 2021,
	month = {Nov.},
	publisher = {Verein zur Forderung des Open Access Publizierens in den Quantenwissenschaften},
	volume = {5},
	pages = {585},
	author = {Pavel Panteleev and Gleb Kalachev},
	title = {Degenerate Quantum {LDPC} Codes With Good Finite Length Performance},
	journal = {Quantum}
}

@INPROCEEDINGS{QuantumTannerCodeszemor,
   title={{Quantum tanner codes}},
  author={Leverrier, Anthony and Z{\'e}mor, Gilles},
  booktitle={2022 IEEE 63rd Annual Symposium on Foundations of Computer Science (FOCS)},
  pages={872--883},
  year={2022},
  organization={IEEE},
  doi={10.1109/FOCS54457.2022.00117}}

@misc{wang2022distanceboundsgeneralizedbicycle,
      title={Distance bounds for generalized bicycle codes}, 
      author={Renyu Wang and Leonid P. Pryadko},
      year={2022},
      eprint={2203.17216},
      archivePrefix={arXiv},
      primaryClass={quant-ph},
      url={https://arxiv.org/abs/2203.17216}, 
}

@misc{delfosse2023choosedecoderfaulttolerantquantum,
      title={How to choose a decoder for a fault-tolerant quantum computer? The speed vs accuracy trade-off}, 
      author={Nicolas Delfosse and Andres Paz and Alexander Vaschillo and Krysta M. Svore},
      year={2023},
      eprint={2310.15313},
      archivePrefix={arXiv},
      primaryClass={quant-ph},
      url={https://arxiv.org/abs/2310.15313}, 
}

@article{Preskill_2025,
   title={Beyond NISQ: The Megaquop Machine},
   volume={6},
   ISSN={2643-6817},
   url={http://dx.doi.org/10.1145/3723153},
   DOI={10.1145/3723153},
   number={3},
   journal={ACM Transactions on Quantum Computing},
   publisher={Association for Computing Machinery (ACM)},
   author={Preskill, John},
   year={2025},
   month=apr, pages={1–7} }

@misc{mckay2013practicalgraphisomorphismii,
      title={Practical graph isomorphism, II}, 
      author={Brendan D. McKay and Adolfo Piperno},
      year={2013},
      eprint={1301.1493},
      archivePrefix={arXiv},
      primaryClass={cs.DM},
      url={https://arxiv.org/abs/1301.1493}, 
}

@misc{pradhan2025lineartimeiterativedecoders,
      title={Linear Time Iterative Decoders for Hypergraph-Product and Lifted-Product Codes}, 
      author={Asit Kumar Pradhan and Nithin Raveendran and Narayanan Rengaswamy and Bane Vasić},
      year={2025},
      eprint={2504.01728},
      archivePrefix={arXiv},
      primaryClass={cs.IT},
      url={https://arxiv.org/abs/2504.01728}, 
}

@misc{muller2025improvedbeliefpropagationsufficient,
      title={Improved belief propagation is sufficient for real-time decoding of quantum memory}, 
      author={Tristan Müller and Thomas Alexander and Michael E. Beverland and Markus Bühler and Blake R. Johnson and Thilo Maurer and Drew Vandeth},
      year={2025},
      eprint={2506.01779},
      archivePrefix={arXiv},
      primaryClass={quant-ph},
      url={https://arxiv.org/abs/2506.01779}, 
}

@article{bravyi2024high,
  title={High-threshold and low-overhead fault-tolerant quantum memory},
  author={Bravyi, Sergey and Cross, Andrew W and Gambetta, Jay M and Maslov, Dmitri and Rall, Patrick and Yoder, Theodore J},
  journal={Nature},
  volume={627},
  number={8005},
  pages={778--782},
  year={2024},
  publisher={Nature Publishing Group UK London}
}

@article{calderbank1996quantum_exists,
title={Good quantum error-correcting codes exist},
  author={Calderbank, A Robert and Shor, Peter W},
  journal={Physical Review A},
  volume={54},
  number={2},
  pages={1098},
  year={1996},
  publisher={APS}
}

@ARTICLE{mackay_quantum,
author={D. J. C. MacKay and G. Mitchison and P. L. McFadden},
journal={IEEE Transactions on Information Theory},
title={Sparse-graph Codes for Quantum Error Correction},
volume={50},
number={10},
pages={2315--2330},
month={Oct.},
year={2004},
doi={10.1109/TIT.2004.834737}
}

@inproceedings{panteleev2021quantumLinearMinDLocalTestable,
title={{Asymptotically good quantum and locally testable classical LDPC codes}},
  author={Panteleev, Pavel and Kalachev, Gleb},
  booktitle={Proceedings of the 54th Annual ACM SIGACT Symposium on Theory of Computing},
  pages={375--388},
  year={2022}
}

@ARTICLE{decim,
      title={{Belief Propagation Decoding of Quantum LDPC Codes with Guided Decimation}}, 
      author={Hanwen Yao and Waleed Abu Laban and Christian Häger and Alexandre Graell i Amat and Henry D. Pfister},
      year={2023},
journal = {arXiv preprint arXiv:2308.13377},
      eprint={2312.10950},
      archivePrefix={arXiv},
      primaryClass={cs.IT}
}

@ARTICLE{05CDEFH,
title={Reduced-Complexity Decoding of {LDPC} Codes},
author={ Chen, J. and Dholakia, A. and Eleftheriou, E. and Fossorier, M.P.C. and Hu, X.-Y.},
journal={IEEE Trans. Commun.},
year={2005},
month={Aug.},
volume={53},
number={8},
pages={1288--1299},
doi={10.1109/TCOMM.2005.852852}
}

@phdthesis{Gottesman97,
	author = {Daniel Gottesman},
	title = {Stabilizer Codes and Quantum Error Correction},
	school = {California Institute of Technology},
	year = {1997},
	eprint = {arXiv:quant-ph/9705052},
	doi = {10.7907/rzr7-dt72}
}

@article{Fuentes_DegeneracyImpact_IEEEAccess21,
author = {Fuentes, Patricio and Etxezarreta, Josu and Crespo, Pedro and Garcia-Frias, Javier},
year = {2021},
month = {Jun.},
pages = {89093--89119},
title = {Degeneracy and Its Impact on the Decoding of Sparse Quantum Codes},
volume = {9},
journal = {IEEE Access},
doi = {10.1109/ACCESS.2021.3089829}
}

@INPROCEEDINGS{iso,
  author={Planjery, Shiva Kumar and Chilappagari, Shashi Kiran and Vasić, Bane and Declercq, David and Danjean, Ludovic},
  booktitle={2010 Information Theory and Applications Workshop (ITA)}, 
  title={Iterative decoding beyond belief propagation}, 
  year={2010},
  volume={},
  number={},
  pages={1-10},
  doi={10.1109/ITA.2010.5454076}}

@misc{amb,
      title={{Ambiguity Clustering: an accurate and efficient decoder for qLDPC codes}}, 
      author={Stasiu Wolanski and Ben Barber},
      year={2024},
      eprint={2406.14527},
      archivePrefix={arXiv},
      primaryClass={quant-ph},
      url={https://arxiv.org/abs/2406.14527}, 
}

@misc{local,
      title={Localized statistics decoding: A parallel decoding algorithm for quantum low-density parity-check codes}, 
      author={Timo Hillmann and Lucas Berent and Armanda O. Quintavalle and Jens Eisert and Robert Wille and Joschka Roffe},
      year={2024},
      eprint={2406.18655},
      archivePrefix={arXiv},
      primaryClass={quant-ph},
      url={https://arxiv.org/abs/2406.18655}, 
}

@misc{window,
      title={{Toward Low-latency Iterative Decoding of QLDPC Codes Under Circuit-Level Noise}}, 
      author={Anqi Gong and Sebastian Cammerer and Joseph M. Renes},
      year={2024},
      eprint={2403.18901},
      archivePrefix={arXiv},
      primaryClass={quant-ph},
      url={https://arxiv.org/abs/2403.18901}, 
}

@INPROCEEDINGS{chytasISIT,
  author={Chytas, Dimitris and Raveendran, Nithin and Vasić, Bane},
  booktitle={2025 IEEE International Symposium on Information Theory (ISIT)}, 
  title={Enhanced Min-Sum Decoding of Quantum Codes with Iteration Dynamics Memory}, 
  year={2025},
  volume={},
  number={},
  pages={1-6},
  doi={10.1109/ISIT63088.2025.11195509}}
\end{document}